\documentclass[lettersize,journal]{IEEEtran}
\usepackage{amsmath,amsfonts}
\usepackage{amsthm}
\usepackage{amssymb}
\usepackage{algorithm}
\usepackage{algpseudocode}
\usepackage{array}
\usepackage[caption=false,font=normalsize,labelfont=sf,textfont=sf]{subfig}
\usepackage{textcomp}
\usepackage{stfloats}

\usepackage{cuted}
\usepackage{url}
\usepackage{verbatim}
\usepackage{graphicx}
\usepackage{cite}
\usepackage{mathrsfs}
\usepackage{textcomp}
\usepackage{footnote}
\usepackage{xcolor}
\usepackage{booktabs}
\usepackage{pifont}
\usepackage{multirow}
\usepackage{makecell} 
\usepackage{algpseudocode}
\usepackage{float}
\usepackage{ulem}
\urldef\baselinecode\url{https://kaimingshen.github.io/doc/asilomar_code.zip}

\newtheorem{theorem}{Theorem}
\newtheorem{proposition}{Proposition}
\newtheorem{lemma}{Lemma}

\newtheorem{remark}{Remark}
\newtheorem{corollary}{Corollary}
\theoremstyle{definition} 

\makeatletter
\renewcommand{\maketag@@@}[1]{\hbox{\m@th\normalsize\normalfont#1}}%

\makeatother
\algnewcommand\algorithmicoutput{\textbf{Output:}}
\algnewcommand\Output{\item[\algorithmicoutput]}

\begin{document}

\title{Rethinking Fractional Programming for Joint Uplink Scheduling and Power Control in Multicell Wireless Networks}

\author{\IEEEauthorblockN{Zihan Jiao, Xinping Yi, \textit{Member, IEEE}, Shi Jin, \textit{Fellow, IEEE}, and Giuseppe Caire, \textit{Fellow, IEEE}}

\thanks{Z. Jiao, X. Yi, and S. Jin are with the School of Information Science and Engineering, Southeast University, Nanjing 210096, China. Email:
\{jiao.zh, xyi, jinshi\}@seu.edu.cn.
}
\thanks{G. Caire is with the Faculty of Electrical Engineering and Computer Science, Technical University of Berlin, Germany. Email: 
caire@tu-berlin.de.}
}


\maketitle

\begin{abstract} This paper investigates the joint uplink scheduling and power control problem in a coordinated multicell wireless network, where at most one single-antenna user is allowed to access the single-antenna base station in each cell simultaneously. The resulting weighted sum-rate (WSR) maximization problem is a mixed discrete-continuous, nonconvex optimization problem that is notoriously difficult to solve directly. Classical fractional programming (FP) methods tackle this problem by leveraging the Lagrangian dual transform (LDT) followed by the quadratic transform (QT), yielding a tractable closed-form solution for scheduling and power control, with the LDT playing a crucial role in handling discrete variables. In this paper, we revisit the LDT from a minorization-maximization (MM) perspective and observe that its induced surrogate is somehow conservative due to the reciprocal-coordinate construction. Motivated by this observation, we propose a novel reciprocal-inversion transform (RIT) that constructs a tighter first-order Taylor expansion lower bound for the logarithmic rate function. The proposed RIT remains fully compatible with the QT, leading to a surrogate-enhanced FP (SEFP) algorithm for joint uplink scheduling and power control. 
The proposed SEFP algorithm retains the desirable per-cell separability of the classical FP framework and admits closed-form updates for the auxiliary variables, scheduling decisions, and transmit powers. Simulation results demonstrate that the SEFP algorithm consistently outperforms the classical FP method and other baselines for different network utilities.

\end{abstract}


\section{Introduction}
Coordinated multicell transmission has been widely recognized as one of the most effective interference management techniques for improving spectral efficiency in dense cellular networks. 
In multicell networks, aggressive frequency reuse makes the achievable user rates in different cells strongly coupled through intercell interference, so that resource allocation decisions can no longer be made independently on a per-cell basis \cite{Gesbert2011, Gesbert2007}. 
This coupling is particularly crucial in the uplink, where each scheduled user may create intercell interference to multiple neighboring base stations, such that the scheduling decision directly changes the network-wide interference pattern \cite{FPPart2}. 
Therefore, the joint design of user scheduling and transmit power control, formulated as a nonconvex mixed discrete-continuous nonlinear program, is a fundamental challenging problem for coordinated multicell uplink systems. Normally, user scheduling determines which users access the shared spectrum, while power control determines how aggressively the scheduled users transmit \cite{YuKwonShin2013}. 

In the literature, there are four pathways to the joint uplink scheduling and power control problem.
The first one is the power control-oriented methods, such as \cite{P2_powerorient_1,P2_powerorient_2}, which treat scheduling implicitly as continuous power optimization that maps zero and positive power allocations to unselected and scheduled users, respectively. 
However, such continuous formulations suffer from a severe drawback known as \textit{premature turning-off} \cite{FPPart2}. 
Specifically, if a communication link is deactivated in the early stages of the iterative optimization process, it can rarely be reactivated in subsequent iterations, which severely reduces the network's utility.
The second one is decomposition-oriented methods.
The main idea of this line of work is to decompose the original mixed discrete-continuous variables problem into more tractable subproblems, such as separate scheduling and power control updates or centralized and distributed resource management components, e.g., \cite{P2_decompositionorient_1,P2_decompositionorient_2,P2_decompositionorient_3}. Although these methods can significantly simplify the original problem, their algorithmic design is often problem-specific, and the resulting performance depends critically on the quality of the adopted decomposition, approximation, or interference treatment strategies.
The third one is artificial intelligence (AI)-aided methods, which have recently attracted increasing attention with the rapid development of AI techniques, e.g., \cite{P2_datadriven1_LORM,P2_datadriven3_gnnbased} to name a few. However, most existing approaches either unfold iterative algorithms with learnable parameters or directly approximate parts of the optimization procedure with neural networks. For this high communication-specific-domain mixed discrete-continuous problem, the solution space of network utility remains highly complex, so AI-aided methods have not yet shown a clear advantage over well-designed iterative algorithms, while their theoretical guarantees of generalization and interpretability are still unclear. 

Of particular relevance is the fourth pathway under the fractional programming (FP) framework, which is also the main focus of the current work.
The seminal FP framework in \cite{FPPart1,FPPart2} established a distinct optimization paradigm for communication network utility optimization problems involving logarithmic rate functions formed as $\Sigma_{i} \log(1+\text{SINR}_i)$.
In \cite{FPPart1}, two FP routes were developed for handling the weighted sum-rate objective. 
The first route is the direct FP approach, which applies the quadratic transform (QT) directly to the fractional signal-to-interference-plus-noise ratio (SINR) term within the logarithmic rate expression, typically leading to iterative updates that involve a series of convex subproblems.
In contrast, the second one, named the closed-form FP, first applies the Lagrangian dual transform (LDT) to move the SINR ratio outside the logarithm, and then uses the QT to decouple the resulting fractional term, yielding iterative closed-form updates for the three optimization variables (i.e., one original variable and two auxiliary variables).
As further clarified in \cite{FPPart2}, for purely continuous optimization problems, these two routes usually give comparable performance, with their main difference lying in computational complexity and update structure. 
However, for optimization problems involving discrete variables inside the logarithmic rate function, the closed-form FP route, especially the use of LDT, becomes indispensable, because it converts the log-ratio objective into a form compatible with subsequent QT-based decoupling (with different cells) and combinatorial scheduling optimization.

One may wonder what role the LDT plays and whether or not it is really indispensable.
Previous attempts include \cite{FPPart3} and \cite{BCA+MM WSR}, which reconsidered FP from an MM perspective and provided a clear theoretical view by showing that the equivalent transforms used in FP can be interpreted as surrogate function constructions.
This perspective makes it possible to examine the intrinsic quality of the lower bounds induced by the FP's transforms. 
Since then, many FP-related algorithmic developments have focused on algorithm acceleration, complexity reduction, and scenario extension, such as FastFP \cite{FastFP}, DeepFP \cite{DeepFP}, EGAT-FP \cite{EGAT-FP}, FP with stochastic CSI \cite{stochastic FP}, and for ISAC designs \cite{ISAC_FP}. 
These works significantly improve the efficiency and applicability of FP, yet they all retain the underlying LDT-plus-QT-based closed-form FP structure. 
Such an LDT-plus-QT framework is effective in high-dimensional continuous-variable problems, such as multicell MIMO beamforming, where the linear precoding and decoding often form a network utility optimization landscape in which FP-type iterations can already attain very strong stationary solutions in practice. 
However, for mixed user scheduling (cf. discrete variables) and power control/beamforming (cf. continuous variables) problems, the solution space becomes more complex and discontinuous due to the addition of discrete variables — the question then arises as to whether the LDT-plus-QT remains the indispensable solution.

Drawing on the insight from \cite{MM tutorial} that a tighter surrogate function can potentially enhance performance, we investigate the feasibility of constructing a more tightly bounded surrogate function.
More specifically, we first rethink the LDT from an MM perspective. Instead of using the Lagrangian dual theory, we reconstruct the LDT-type surrogate by variable substitution and tangent-envelope construction. 
Therefore, this paper improves the closed-form FP framework at the transform level by proposing a reciprocal-inversion transform (RIT), which preserves the desirable QT-compatibility and per-cell separability while providing a tighter surrogate for the logarithmic rate function. Specifically, the main contributions of this paper are summarized as follows.
\begin{itemize}
    \item We revisit the classical LDT from an MM perspective and identify its structural conservativeness. To overcome this limitation, we propose the RIT and combine it with the QT, yielding an equivalent SEFP reformulation with a tighter MM surrogate than the classical LDT-plus-QT formulation.
    \item We apply the proposed RIT-plus-QT framework to the joint uplink scheduling and power control problem in coordinated multicell networks. The resulting SEFP algorithm preserves per-cell separability and closed-form updates, while introducing a more adaptive scheduling metric than the classical FP algorithm, thereby facilitating more effective per-iteration scheduling and power-control updates.
    \item We conduct extensive simulations under multiple utility metrics and SNR regimes. The proposed SEFP algorithm consistently outperforms classical FP and other representative baselines under all considered settings, demonstrating robust gains and state-of-the-art empirical performance for the considered problem.
\end{itemize}

The remainder of this paper is organized as follows. 
Section II reviews the generic weighted sum-of-logarithms maximization problem and the classical closed-form FP framework based on the LDT-plus-QT, and then revisits the LDT from an MM surrogate perspective.
Section III develops the proposed RIT, presents its constructive derivation, and establishes the resulting RIT-plus-QT formulation as an SEFP framework.
Section IV applies the proposed framework to the joint uplink scheduling and power control problem in multicell networks, deriving closed-form updates and analyzing convergence and computational complexity.
Section V presents simulation results, and Section VI concludes the paper.

\textit{Notation:} Throughout this paper, scalar quantities are denoted by italic letters, while vectors, tuples, or collections of variables are denoted either by boldface lowercase letters or by explicitly defined collections, calligraphic letters denote sets, e.g., $\mathcal{X}$. 
The symbols $\mathbb{R}$, $\mathbb{R}{+}$ and $\mathbb{R}{++}$ denote the sets of real, nonnegative and positive real numbers, respectively.
For objective functions obtained through equivalent transforms, multiple arguments are separated by commas, e.g., $F_r(\mathbf{x}, \boldsymbol{\alpha})$ and $F_{rq}(\mathbf{x}, \boldsymbol{\alpha}, \mathbf{y})$, where all arguments are original optimization variables or auxiliary optimization variables of the equivalent reformulated problem.
For MM surrogate functions, a semicolon is used to separate the variable being optimized from the fixed reference point used to construct the surrogate, e.g., $\underline{F}_{\text{RIT}}(\mathbf{x}; \mathbf{x}^{(t)})$.
The superscript $(t)$ denotes the iteration index, while $(\cdot)^\star$ denotes an optimal solution or optimal auxiliary-variable update. 
The notation $|\cdot|$ denotes the cardinality of a set. The logarithm $\log(\cdot)$ is taken to be natural unless otherwise specified.

\section{Preliminaries}
\subsection{Weighted Sum-of-logarithms Problem and Classical Fractional Programming Framework}
Consider the following generic weighted sum-of-logarithms maximization problem
\begin{equation}\label{w_sol_p}
\max_{\mathbf{x} \in \mathcal{X}} \quad F(\mathbf{x}) \triangleq \sum_{m=1}^{M} \omega_m \log \left( 1 + \frac{A_m(\mathbf{x})}{B_m(\mathbf{x})} \right),\end{equation}
where $\mathbf{x}$ denotes the collection of optimization variables, 
$\mathcal X$ is the feasible set. 
For each index $m$, $\omega_m$ is a nonnegative weight, $A_m(\mathbf{x})\geq 0$ denotes the numerator function, and $B_m(\mathbf{x})> 0$ denotes the denominator function. 

In wireless communications, by treating interference as noise at the receiver, the signal-to-interference-and-noise ratio (SINR) is usually defined as $\text{SINR}=A_m(\mathbf{x})/B_m(\mathbf{x})$, where $A_m(\mathbf{x})$ and $B_m(\mathbf{x})$ represent the desired-signal power and the interference-plus-noise power, respectively. As such, \eqref{w_sol_p} is naturally a WSR maximization problem.

Fractional programming (FP) provides an effective reformulation framework for optimization problems involving such ratio structures.
A critical tool in the classical FP framework is the QT delineated in the following lemma, which decouples the numerator from the denominator in a fractional term.
\begin{lemma} (Quadratic Transform \cite[Theorem 1]{FPPart1}) \label{QTlemma}
For a ratio $A(\mathbf{x})/B(\mathbf{x})$, where $A(\mathbf{x}) \ge 0$ and $B(\mathbf{x}) > 0$, the QT gives
\begin{equation}
\frac{A(\mathbf{x})}{B(\mathbf{x})} = \max_{y \in \mathbb{R}} \left[ 2y\sqrt{A(\mathbf{x})} - y^2 B(\mathbf{x}) \right],  
\end{equation}
where the optimal auxiliary variable is
\begin{equation}
y^\star = \frac{\sqrt{A(\mathbf{x})}}{B(\mathbf{x})}.
\end{equation}
\end{lemma}
The QT is especially useful for multi-ratio problems because it preserves the objective value equivalence after optimizing over the auxiliary variables. 
Therefore, a sum of ratios can be transformed term-by-term into a more tractable form, which enables efficient iterative optimization.

For WSR problems involving $\log (1+\text{SINR})$, classical FP methods follow two routes. 
The first route, named as the \textit{direct FP}, applies the QT directly to the SINR term inside the logarithm, which is more suitable for continuous optimization problems and leads to iterative updates through convex subproblems.
In this work, however, we focus specifically on the second route, namely the \textit{closed-form FP}. 
This route first leverages the LDT to move the SINR ratio outside the logarithm, and then applies the QT to the resulting formula with the ratio term. 
It is particularly effective for problems involving discrete variables, because it converts the logarithmic ratio objective into a form compatible with subsequent decoupling and combinatorial discrete variables' optimization.

\subsection{Closed-Form FP = LDT + QT} \label{Section2Subsection2}
We first review the LDT in the following lemma.
\begin{lemma} (Lagrangian Dual Transform \cite[Theorem 3]{FPPart2}) \label{LDTlemma}
By taking the LDT, the weighted sum-of-logarithms objective $F(\mathbf{x})$ in \eqref{w_sol_p} can be equivalently reformulated as
\begin{equation}
\max_{\substack{\mathbf{x} \in \mathcal{X}, \ \boldsymbol{\gamma} \in \mathbb{R}_+^M}}\quad F_\ell(\mathbf{x}, \boldsymbol{\gamma}), \end{equation}
where $\boldsymbol{\gamma}$ denotes the LDT-induced auxiliary variables tuple $(\gamma_1, \gamma_2, \dots, \gamma_M)$, and
\begin{equation}\label{eq:LDT_objective}
F_\ell(\mathbf{x}, \boldsymbol{\gamma})\triangleq\sum_{m=1}^M\omega_m\left[ \log(1 + \gamma_m) - \gamma_m + \frac{(1 + \gamma_m)A_m(\mathbf{x})}{A_m(\mathbf{x}) + B_m(\mathbf{x})} \right].
\end{equation}
For any fixed $\mathbf{x}$, the optimal LDT auxiliary variables admit the closed-form update
\begin{equation}\gamma_m^\star = \frac{A_m(\mathbf{x})}{B_m(\mathbf{x})}, \quad \forall m.
\end{equation}
\end{lemma}

The key effect of the LDT is that the ratio term is moved outside the logarithm. 
For fixed $\boldsymbol{\gamma}$, the $\mathbf{x}$-dependent fractional part in $F_\ell(\mathbf{x}, \boldsymbol{\gamma})$ becomes
$$\frac{\omega_m(1 + \gamma_m)A_m(\mathbf{x})}{A_m(\mathbf{x}) + B_m(\mathbf{x})},\quad m=1,...,M, $$
which is directly compatible with the QT.
Therefore, for a fixed $\boldsymbol{\gamma}$, the QT can be directly applied to each fractional term in $F_\ell(\mathbf{x}, \boldsymbol{\gamma})$.
Thus, combining the LDT and the QT, \eqref{w_sol_p} can be further equivalently reformulated as
\begin{equation} \label{Flq} 
\begin{split}
 \max_{\substack{ \mathbf{x} \in \mathcal{X}\\ \boldsymbol{\gamma} \in \mathbb{R}_{+}^M \\ \mathbf{y} \in \mathbb{R}^M}}&\ F_{\ell q}(\mathbf{x}, \boldsymbol{\gamma}, \mathbf{y})\triangleq\sum_{m=1}^{M} \bigg[  2y_m \sqrt{\omega_m(1 + \gamma_m)A_m(\mathbf{x})}  \\ & -y_m^2 \big(A_m(\mathbf{x}) + B_m(\mathbf{x})\big)+\omega_m\big(\log(1 + \gamma_m) - \gamma_m\big) \bigg], 
\end{split}
\end{equation}
where $\mathbf{y}$ denotes the QT-induced auxiliary variables tuple $({y}_1, {y}_2, \dots, {y}_M)$, with the optimal closed-form update
\begin{equation}
y_m^\star = \frac{\sqrt{\omega_m(1 + \gamma_m)A_m(\mathbf{x})}}{A_m(\mathbf{x}) + B_m(\mathbf{x})},\quad m=1,...,M.
\end{equation}
The classical closed-form FP framework updates the QT auxiliary variable $\mathbf{y}$, the LDT auxiliary variable $\boldsymbol{\gamma}$, and the original optimization variable $\mathbf{x}$ alternatively, since $F_{\ell q}(\mathbf{x}, \boldsymbol{\gamma}, \mathbf{y})$ is convex respect to each variable when the other two are fixed.
Notably, for mixed continuous-discrete variables optimization problems, this LDT$+$QT structure is especially valuable.
With the auxiliary variables fixed, the transformed objective becomes decomposable with respect to different discrete decision blocks, so that the discrete variables can be updated independently or through simplified combinatorial subproblems.
This property is essential for extending the closed-form FP framework from purely continuous optimization to mixed continuous-discrete settings.

\subsection{Structural Limitation of the LDT} \label{Section2subsection3}
As previously shown, the LDT is an effective transform for converting the weighted sum-of-logarithms objective into a QT-compatible fractional form.
Further to the Lagrangian duality argument in \cite{FPPart2}, recent studies in \cite{BCA+MM WSR} and \cite{FPPart3} deepen the understanding of the LDT by interpreting it as a means of constructing surrogate functions under the MM framework. 
Nevertheless, both \cite{BCA+MM WSR} and \cite{FPPart3} rely entirely on established LDT results to construct the surrogate function.

In what follows, we revisit the LDT from a constructive perspective, employing two standard surrogate construction techniques—variable substitution and first-order tangent minorization—both of which are foundational in the MM framework \cite{MM tutorial}.
Remarkably, such a perspective reveals that the LDT can be equivalently interpreted as a first-order tangent lower bound of a surrogate function constructed with variable substitution of the original logarithmic function.
Whilst this construction maintains the desired QT-compatible fractional structure, it exposes several intrinsic restrictions of the resulting surrogate constructions.

\subsubsection{Revisiting LDT from a Surrogate Perspective}
To make the discussion precise, we first introduce the ratio
\begin{equation}\label{rmx_definition}
r_m(\mathbf{x}) \triangleq \frac{A_m(\mathbf{x})}{B_m(\mathbf{x})}, \quad m = 1, \dots, M.\end{equation}
Since $A_m(\mathbf{x}) \geq 0$ and $B_m(\mathbf{x}) > 0$, we have $r_m(\mathbf{x}) \geq 0$. For the purpose of single-ratio analysis, define the unweighted logarithmic sub-function
\begin{equation} \label{fr}
f(r) \triangleq \log(1 + r), \quad r \geq 0.
\end{equation}
Then, the $m$-th logarithmic term in $F(\mathbf{x})$ can be written as $f(r_m(\mathbf{x}))$. 
In the following discussion, $r$ denotes a generic ratio variable, while $\bar{r}$ denotes a fixed reference point.

We now reconstruct the LDT surrogate from a variable-substitution perspective. Consider the substitution
\begin{equation}z = \frac{1}{1 + r}.\end{equation}
Then $f(r)=  -\log z.$ 
Since $-\log z$ is convex over $z>0$, its first-order tangent approximation at any reference point $\bar z$ gives a global lower bound 
\begin{equation} \label{1storder_app_oflogz}
    -\log z \geq -\log \bar z - \frac{1}{\bar z}(z - \bar z).
\end{equation}
Substituting $z = {1}/{(1+r)}$ and $\bar z = {1}/{(1+\bar r)}$ into \eqref{1storder_app_oflogz}, we obtain the LDT surrogate of $f(r)$, expressed as
\begin{equation} \label{surrogate_of_fr}
\log(1 + r) \geq \log(1 + \bar r) - \bar r + \frac{(1 + \bar r)r}{1 + r}\triangleq \ell_{\mathrm{LDT}}(r;\bar r).
\end{equation}

We now return to the original objective $F(\mathbf{x})$, at the $t$-th MM iteration, let $\mathbf{x}^{(t)}$ denote the current iterate, and the reference ratio for the $m$-th logarithmic term be
\begin{equation} \label{ratio for t-th iter}
\bar r_m^{(t)}\triangleq r_m(\mathbf{x}^{(t)}) = \frac{A_m(\mathbf{x}^{(t)})}{B_m(\mathbf{x}^{(t)})}.
\end{equation}
Applying the above single-ratio surrogate to each logarithmic term yields the LDT-induced MM surrogate of $F(\mathbf{x})$, i.e.,
\begin{equation} \label{F_LDT_surrogate}
\begin{aligned}
\underline{F}_{\text{LDT}}&(\mathbf{x}; \mathbf{x}^{(t)}) \triangleq \sum_{m=1}^{M} \omega_m \ell_{\text{LDT}}\left(r_m(\mathbf{x}); \bar{r}_m^{(t)}\right) \\
&= \sum_{m=1}^{M} \omega_m \left[ \log(1 + \bar{r}_m^{(t)}) - \bar{r}_m^{(t)} + \frac{(1 + \bar{r}_m^{(t)})A_m(\mathbf{x})}{A_m(\mathbf{x}) + B_m(\mathbf{x})} \right],
\end{aligned}
\end{equation}
which satisfies
\begin{equation}
\underline{F}_{\mathrm{LDT}} ( \mathbf{x} ; \mathbf{x}^{(t)} ) \leq F(\mathbf{x}), \quad \forall \mathbf{x} \in \mathcal{X},
\end{equation}
and it is tight at the current iterate
\begin{equation}
\underline{F}_{\mathrm{LDT}} \left.( \mathbf{x} ; \mathbf{x}^{(t)} )\right|_{\mathbf{x}=\mathbf{x}^{(t)}} = \left.F(\mathbf{x})\right|_{\mathbf{x}=\mathbf{x}^{(t)}}.
\end{equation}  

\subsubsection{Structural Limitations of the LDT Surrogate}
The above constructive interpretation, however, shows that the tangent operation is not performed in the original coordinate $r$, but in the substituted coordinate $z={1}/{(1+r)}.$
This mapping compresses the entire interval $r\in[0,\infty)$ into $z\in(0,1]$. Although this compression is essentially leading to the QT-compatible fractional structure, it also makes the resulting surrogate conservative. The following proposition summarizes three representative structural limitations by comparing $f(r)$ with its corresponding LDT surrogate $\ell_{\mathrm{LDT}}(r;\bar r)$.
\begin{proposition} (Structural limitations of the LDT surrogate)
For any fixed reference point $\bar r>0$, we have
 \begin{enumerate}
     \item $\ell_{\mathrm{LDT}}(r;\bar r)$ does not preserve the exact boundary behavior of $f(r)$ at $r=0$. Since,
     $$f(r)\left.\right|_{r=0} = 0.$$
    whereas
    $$
    \ell_{\mathrm{LDT}}(0; \bar{r}) = \log(1+\bar{r}) - \bar{r} < 0.$$  
    \item Although $\ell_{\mathrm{LDT}}(r;\bar r)$ and $f(r)$ are first-order consistent, its local curvature is more conservative. 
    Specifically,
    $$\ell''_{\text{LDT}}(\bar{r}; \bar{r}) = 2f''(\bar{r})<0.$$
    \item The LDT surrogate exhibits a growth mismatch in the high-ratio regime. As $r \to \infty$, 
    $$f(r) \to \infty. $$
    whereas,
    $$
    \ell_{\mathrm{LDT}}(r; \bar{r}) \to \log(1 + \bar{r}) + 1 < \infty. 
    $$
 \end{enumerate}
\end{proposition}
\begin{proof}
    See Appendix A.
\end{proof}

It is worth noting that the above limitations does not invalidate LDT, which remains an elegant and useful transformation to convert each logarithmic ratio term into a tractable QT-compatible fractional form. 
However, Proposition 1 shows that the LDT-induced surrogate construction is intrinsically tied to a particular reciprocal-coordinate substitution $z=1/(1+r)$. 
Once this construction is adopted, the resulting fractional structure is also fixed.
Specifically, the $\mathbf{x}$-dependent ratio structure in \eqref{F_LDT_surrogate} is always
\begin{equation} \label{LDT_structure}
\frac{A_m(\mathbf{x})}{A_m(\mathbf{x})+B_m(\mathbf{x})},\end{equation}
where $A_m(\mathbf{x})$ and $B_m(\mathbf{x})$ enter the denominator with fixed equal weighting. 
The reference ratio $\bar r_m^{(t)}$ only changes the outer coefficient and the additive constant of the surrogate, but cannot adjust the relative weighting between $A_m(\mathbf{x})$ and $B_m(\mathbf{x})$ inside the denominator.
This structural restriction motivates us to develop a more flexible surrogate construction in the next section.

\section{Proposed Reciprocal-Inversion Transform and its Combining With the Quadratic Transform} \label{Section3}
In this section, we develop a novel Reciprocal-Inversion Transform (RIT) and establish its compatibility with the QT.
Specifically, in Section~\ref{Section3subsection1}, we first present the RIT as an equivalent transform for the weighted sum-of-logarithms objective. 
Then, in Section~\ref{Section3subsection2}, we reinterpret the RIT from the MM perspective through the corresponding Reciprocal-Inversion surrogate function. 
By detailing its constructive derivation, we clarify how the proposed surrogate is obtained and how it remedies the structural limitations of the LDT surrogate discussed in Section~\ref{Section2subsection3}. 
Finally, in Section~\ref{Section3subsection3}, we combine the proposed RIT with the QT and prove that the resulting RIT-plus-QT surrogate is strictly tighter than the LDT-plus-QT surrogate while preserving the QT-compatible fractional structure required by the closed-form FP framework.

\subsection{The Reciprocal-Inversion
Transform (RIT)} \label{Section3subsection1}
To simplify the notation, for any $v>0$, define
\begin{equation} \label{operator_b}
b(v)\triangleq (1 + v)\log(1+v) - v,
\end{equation} 
and
\begin{equation} \label{operator_c}
c(v)\triangleq \left( 1 + v \right) \log^2(1+v).
\end{equation}

\begin{theorem} (Reciprocal-Inversion Transform) \label{RIT}
The weighted sum-of-logarithms problem \eqref{w_sol_p} can be equivalently reformulated as
\begin{equation} \label{Fr_whole_problem}
\begin{aligned}
   &\max \quad
   F_r(\mathbf{x},\boldsymbol{\alpha}) \\
   &\quad\mathrm{s.t.} \quad \mathbf{x} \in \mathcal{X}, \quad{\boldsymbol{\alpha} \in \mathbb{R}_{++}^M},
\end{aligned}
\end{equation}
where $\boldsymbol{\alpha}=(\alpha_1, \alpha_2, \dots, \alpha_M)$ denotes the auxiliary variables tuple introduced by the RIT, and 
\begin{equation}\label{eq:RIT_objective}
F_r(\mathbf{x},\boldsymbol{\alpha})\triangleq\sum_{m=1}^{M} \omega_m \frac{c(\alpha_m) A_m(\mathbf{x})}{\alpha_m^2 B_m(\mathbf{x}) + b(\alpha_m) A_m(\mathbf{x})}.
\end{equation}

For any fixed $\mathbf{x}$, when $A_m(\mathbf{x})/B_m(\mathbf{x})>0$, the optimal auxiliary variable is given in closed-form by
\begin{equation} \label{optimal_RIT_av_update}
    \alpha_m^\star = \frac{A_m(\mathbf{x})}{B_m(\mathbf{x})}, \quad m = 1, \dots, M.
\end{equation}
When $A_m(\mathbf{x})/B_m(\mathbf{x}) = 0$, the corresponding $m$-th term of $F_r(\mathbf{x},\boldsymbol{\alpha})$ equals zero for any $\alpha_m > 0$, and hence exactly matches the $m$-th term of $F(\mathbf{x})$.

\end{theorem}
The two problems are equivalent in the sense that $\mathbf{x}$ is the solution to \eqref{w_sol_p} if and only if it is the solution to \eqref{Fr_whole_problem}, so that optimizing the RIT-reformulated objective $F_r(\mathbf{x},\boldsymbol{\alpha})$ in \eqref{Fr_whole_problem} over $\boldsymbol{\alpha}$ exactly recovers the original objective $F(\mathbf{x})$ in \eqref{w_sol_p}.

\begin{proof}
Observing that $F_r(\mathbf{x},\boldsymbol{\alpha})$ is separable with respect to $\{\alpha_m\}_{m=1}^M$ and differentiable over $\boldsymbol{\alpha}$ when $\mathbf{x}$ is held fixed, so each auxiliary variable can be optimally determined in an independent way.

Following the ratio definition $r_m(\mathbf{x})$ in \eqref{rmx_definition}, the case $r_m(\mathbf{x})=0$ is well-stated in Theorem~\ref{RIT}.
To proof the case $r_m(\mathbf{x})>0$, by setting the derivative of the $m$-th term of $F_r(\mathbf{x},\boldsymbol{\alpha})$ with respect to $\alpha_m$ to zero, we obtain \eqref{optimal_RIT_av_update}. Substituting this $\alpha_m^\star$ back into the $m$-th term of the RIT-reformulated objective $F_r(\mathbf{x},\boldsymbol{\alpha})$ gives
\begin{equation}
    \frac{c(r_m(\mathbf{x}))r_m(\mathbf{x})}{r_m^2(\mathbf{x}) + b(r_m(\mathbf{x}))r_m(\mathbf{x})} = \log(1 + r_m(\mathbf{x}))=f(r_m(\mathbf{x})).
\end{equation}
Multiplying by $\omega_m$ and summing over all $m=1,...,M$ exactly recovers the original weighted sum-of-logarithms objective $F(\mathbf{x}).$ 
This completes the proof.\end{proof}

\subsection{Constructive Derivation} \label{Section3subsection2}
To provide insights on how the RIT in Theorem~\ref{RIT} is obtained, 
we leverage the same analytical methodology as in Section~\ref{Section2subsection3}, and revisit $F(\mathbf{x})$ from the surrogate construction perspective under the MM framework.
Following the definition $r_m(\mathbf{x})$ in \eqref{rmx_definition} and the unweighted logarithmic sub-function $f(r)$ in \eqref{fr}, let $r$ denote a generic ratio variable, and $\bar{r} > 0$ denote a fixed reference point.

First, introduce a different variable substitution operation as
$$ \tau = \frac{1}{r},\quad r>0.$$
Note that the following construction is first carried out for $r>0$, and the resulting surrogate will then be extended to $r=0$ by continuity arguments.

    It is worth noting that,
    in contrast to the LDT variable substitution $z=1/(1+r)$, which directly couples the numerator and denominator through the shifted ratio $1+r$, the RIT variable substitution $\tau=1/r$ inverts the original ratio itself.
    As such, when substituting $r_m(\mathbf{x})=A_m(\mathbf{x})/B_m(\mathbf{x})$, the RIT does not force $A_m(\mathbf{x})$ and $B_m(\mathbf{x})$ combined with fixed equal weighting. 
    From an SINR perspective, the desired-signal power and the interference-plus-noise power are therefore not necessarily coupled.
    This provides the structural flexibility to overcome the LDT surrogate's limitation discussed in Section~\ref{Section2subsection3}.

Then, inspired by the surrogate construction in \cite{MM tutorial} that manipulates the objective function to utilize the hidden concavity, we consider the following reciprocal function
\begin{equation}
h(\tau) = \frac{1}{\log(1 + 1/\tau)}, \quad \tau > 0.
\end{equation}
\begin{proposition} \label{h(t)_pro}
    $h(\tau)$ is strictly concave on $\tau>0.$
\end{proposition}
\begin{proof}
    The second derivative of $h(\tau)$ on the domain of definition is
\begin{equation}
h''(\tau) = \frac{2 - (2\tau + 1) \log(1 + 1/\tau)}{\tau^2(\tau + 1)^2 [\log(1 + 1/\tau)]^3}.
\end{equation}
Since
\begin{equation}
(2\tau + 1) \log(1 + 1/\tau) > 2, \quad \tau > 0,
\end{equation}
we have
\begin{equation}
h''(\tau) < 0.
\end{equation}
Therefore, $h(\tau)$ is strictly concave on $\tau > 0$. 
\end{proof}
By Proposition~\ref{h(t)_pro}, the first-order tangent of $h(\tau)$ gives a global upper bound.\footnote{We adhere to the core rationale of constructing tangent operation, as the resulting surrogate function is desired to retain a form amenable to the QT framework. For $\log$-type functions, a first-order tangent approximation yields the most natural formulation to achieve this objective.}
Since $f(r) = 1/h(\tau)$, this upper bound of $h(\tau)$ can be converted into a lower bound surrogate of $f(r)$.
Let the reference point in the reciprocal-inversion coordinate be $\bar{\tau}=1/\bar r$.
Then, for any $\tau>0$, we have
\begin{equation}
    h(\tau) \le h(\bar{\tau}) + h'(\bar{\tau})(\tau - \bar{\tau}).
\end{equation}
Substituting $\tau =1/r$ and $\bar{\tau}=1/\bar r$, and using $f(r)=\log(1+r)$, we obtain
\begin{equation}
    \frac{1}{f(r)} \le \frac{1}{f(\bar{r})} + \frac{\bar{r}^2}{(1 + \bar{r})f^2(\bar{r})} \left( \frac{1}{r} - \frac{1}{\bar{r}} \right).
\end{equation}

Since both sides are positive, taking reciprocals reverses the inequality and yields
\begin{equation}
    f(r) \geq  \frac{c(\bar{r})r}{\bar{r}^2 + b(\bar{r})r}, \quad r > 0,
\end{equation}
where the $b(\cdot)$ and $c(\cdot)$ are defined in \eqref{operator_b} and \eqref{operator_c}, respectively.
Moreover, the right-hand side is well defined at $r=0$ by continuity and equals zero, which coincides with $f(0)$.
Therefore, the above lower bound naturally extends to the whole domain $r\ge0.$

\begin{corollary} (The RIT Surrogate for $f(r)$) \label{f(rm)RITsurrogate}
 Define 
 \begin{equation}
\ell_{\textnormal{RIT}}(r; \bar{r}) \triangleq \frac{c(\bar{r})r}{\bar{r}^2 + b(\bar{r})r},\ r \geq 0,\bar r>0.
\end{equation}
It follows that  $ \ell_{\textnormal{RIT}}(r; \bar{r})$ is a first-order tight lower bound surrogate of $f(r)$ over $r \geq 0$.
\end{corollary}
\begin{proof}
    See Appendix~\ref{proof_of_corr1}. 
\end{proof}
We now return to the original weighted sum-of-logarithms problem in \eqref{w_sol_p}. Let $\mathbf{x}^{(t)}$ denote the current iterate, and the reference ratio for the $m$-th logarithmic term be same with \eqref{ratio for t-th iter}. Applying Corollary \ref{f(rm)RITsurrogate} to the $m$-th logarithmic term of $F(\mathbf{x})$ yields
\begin{equation}
\log \left( 1 + \frac{A_m(\mathbf{x})}{B_m(\mathbf{x})} \right) \ge \frac{c(\bar{r}_m^{(t)}) A_m(\mathbf{x})}{(\bar{r}_m^{(t)})^2 B_m(\mathbf{x}) + b(\bar{r}_m^{(t)}) A_m(\mathbf{x})}.
\end{equation}

Therefore, $F(\mathbf{x})$ admits the following RIT-induced MM lower-bound surrogate
\begin{equation} \label{F_RIT_surrogate}
\begin{aligned}
\underline{F}_{\text{RIT}}(\mathbf{x} ; \mathbf{x}^{(t)})  &\triangleq \sum_{m=1}^{M} \omega_m \ell_{\text{RIT}}\left(r_m(\mathbf{x}); \bar{r}_m^{(t)}\right) \\
&=\sum_{m=1}^{M} \omega_m \frac{c(\bar{r}_m^{(t)}) A_m(\mathbf{x})}{(\bar{r}_m^{(t)})^2 B_m(\mathbf{x}) + b(\bar{r}_m^{(t)}) A_m(\mathbf{x})},
\end{aligned}
\end{equation}
which satisfies
\begin{equation}
\underline{F}_{\text{RIT}}(\mathbf{x} ; \mathbf{x}^{(t)}) \le F(\mathbf{x}), \quad \forall \mathbf{x} \in \mathcal{X},
\end{equation}
and the equality holds at the current iterate, i.e.,
\begin{equation}
\underline{F}_{\text{RIT}}\left.( \mathbf{x} ; \mathbf{x}^{(t)} )\right|_{\mathbf{x}=\mathbf{x}^{(t)}} = \left.F(\mathbf{x})\right|_{\mathbf{x}=\mathbf{x}^{(t)}}.
\end{equation}
As such, it demonstrates that the RIT in Theorem~\ref{RIT} originates from a first-order tight MM surrogate constructed in the reciprocal-inversion coordinate.

Compared with the LDT surrogate of $F(\mathbf{x})$ in \eqref{F_LDT_surrogate}, the RIT surrogate \eqref{F_RIT_surrogate} gives the fractional structure
\[
\frac{c(\bar{r}_m)A_m(\mathbf{x})}{\bar{r}_m^2 B_m(\mathbf{x}) + b(\bar{r}_m)A_m(\mathbf{x})},
\]
where the \textit{unequal weighting} of $A_m(\mathbf{x})$ and $B_m(\mathbf{x})$ inside the denominator is controlled by the reference ratio $\bar{r}_m$. This is fundamentally different from the LDT-induced structure in \eqref{LDT_structure},
remedying the structural limitations of the LDT surrogate while retaining a QT-compatible fractional form. Interestingly, such an unequal weighting will be shown to be particularly effective in the joint uplink scheduling and power control problem as in Section~\ref{Section4}.


\subsection{Surrogate-Enhanced FP = RIT + QT} \label{Section3subsection3}
By combining the proposed RIT with the QT, we propose a new FP framework, termed as Surrogate-Enhanced FP (SEFP), where the RIT retains the advantage of the LDT with each logarithmic ratio converted into a fractional form that can be further handled by the QT, as reviewed in Section~\ref{Section2Subsection2}. 
Therefore, the RIT remains fully compatible with the closed-form FP solution framework. 

Specifically, from the RIT-equivalent formulation in \eqref{Fr_whole_problem} and \eqref{eq:RIT_objective}, for any fixed $\boldsymbol{\alpha}$, the $m$-th term of $F_r(\mathbf{x},\boldsymbol{\alpha})$ is a valid fractional term with a nonnegative numerator and a positive denominator. Therefore, the QT in Lemma~\ref{QTlemma} can be directly applied to each RIT-induced fractional term. This leads to the following RIT+QT equivalent reformulation.
\begin{theorem}(Surrogate-Enhanced FP) \label{theorem_RIT_QT}
The generic weighted sum-of-logarithms problem \eqref{w_sol_p} is equivalent to
\begin{equation} \label{Frq}
\begin{aligned}
  & \max \quad
   F_{rq}(\mathbf{x}, \boldsymbol{\alpha}, \mathbf{y}) = \sum_{m=1}^{M}  \Bigg[ 2y_m \sqrt{\omega_m c(\alpha_m) A_m(\mathbf{x})}  \\
& \quad\quad\quad\quad\quad\quad\quad - y_m^2 \Big( \alpha_m^2 B_m(\mathbf{x}) + b(\alpha_m) A_m(\mathbf{x}) \Big) \Bigg], \\
   &\quad\mathrm{s.t.} \quad \mathbf{x} \in \mathcal{X}, \quad{\boldsymbol{\alpha} \in \mathbb{R}_{++}^M},\quad\mathbf{y}\in \mathbb{R}^M.
\end{aligned}
\end{equation}
Here, $\mathbf{y} = (y_1, \dots, y_M)$ denotes the auxiliary variables tuple introduced by the QT. For any fixed $(\mathbf{x}, \boldsymbol{\alpha})$, the optimal QT auxiliary variable is given in closed form by
\begin{equation}
y_m^\star = \frac{\sqrt{\omega_m c(\alpha_m) A_m(\mathbf{x})}}{\alpha_m^2 B_m(\mathbf{x}) + b(\alpha_m) A_m(\mathbf{x})}, \quad m = 1, \dots, M.
\end{equation}
The auxiliary variables $\boldsymbol{\alpha}$ are the identical RIT-induced auxiliary variables as those in Theorem~\ref{RIT}.
\end{theorem}

\begin{proof}
Since the RIT and the QT are rigorously equivalent transforms,  \eqref{Frq} can be directly obtained by performing the QT on \eqref{Fr_whole_problem}, that is,
\begin{equation}
\max_{\mathbf{y} \in \mathbb{R}^M} F_{rq}(\mathbf{x}, \boldsymbol{\alpha}, \mathbf{y}) = F_r(\mathbf{x}, \boldsymbol{\alpha}),
\end{equation}
and hence,
\begin{equation}
\max_{\boldsymbol{\alpha} \in \mathbb{R}_{++}^M, \mathbf{y} \in \mathbb{R}^M} F_{rq}(\mathbf{x}, \boldsymbol{\alpha}, \mathbf{y}) = F(\mathbf{x}).
\end{equation}
Therefore, applying the RIT followed by the QT provides an equivalent reformulation of the original weighted sum-of-logarithms problem. The proof is complete.
\end{proof}
\begin{proposition} (Tightness of the SEFP Surrogate)\label{RITtightproposition}
 From the MM perspective, when the respective auxiliary variables are optimally updated at the same current iterate, the RIT+QT formulation yields a uniformly tighter surrogate of the original weighted sum-of-logarithms objective $F(\mathbf{x})$ in \eqref{w_sol_p} than the classical LDT+QT formulation.
\end{proposition}
\begin{proof}
See Appendix~\ref{proof_of_RITtight}.
\end{proof}

Theorem~\ref{theorem_RIT_QT} and Proposition~\ref{RITtightproposition} clarify the role of the RIT in the proposed SEFP framework, which is competitive to the closed-form FP.
Specifically, Theorem~\ref{theorem_RIT_QT} shows the RIT preserves the QT-compatible fractional structure. 
Proposition~\ref{RITtightproposition} further establishes that, under the same MM interpretation, the RIT+QT  provides a more faithful minorization surrogate of the original weighted sum-of-logarithms objective than the LDT+QT counterpart in classical closed-form FP.

\section{Surrogate-Enhanced FP for Joint Uplink Scheduling and Power Control}
\label{Section4}

\subsection{System Model and Problem Formulation}
We consider the uplink of a coordinated multicell wireless network that consists of a set of base 
stations (BSs) $\mathcal{B}$ with single antenna equipped at each BS. Let $\mathcal{K}$ denote the user set in the network and $\mathcal{K}_i$ denote the set of users associated with BS $i \in \mathcal{B}$. 
In each scheduling interval, at most one user is scheduled in each cell with a single-antenna BS. The scheduling decision made by BS $i$ is denoted by
\begin{equation}
s_i \in \mathcal{K}_i \cup \{\varnothing\},
\end{equation}
where $s_i = \varnothing$ means that no user is scheduled in cell $i$, indicating that the cell $i$ is switched off. 
Let $p_k$ denote the transmit power of user $k$, subject to
\begin{equation}
0 \leq p_k \leq P_{\max}.
\end{equation}

The uplink channel coefficient from user $k$ to BS $i$ is denoted by $h_{i,k}$, and the noise power is $\sigma^2$. 
For a given scheduling vector $\mathbf{s} = (\{s_i\}_{i \in \mathcal{B}})$ and power vector $\mathbf{p}=(\{p_k\}_{k \in \mathcal{K}})$, by treating interference as noise at the receiver, the WSR maximization problem with joint uplink scheduling and power control is then formulated as
\begin{subequations} \label{eq:main_problem} 
\begin{align}
\max_{\mathbf{s},\mathbf{p}} \quad & W_{o}(\mathbf{s},\mathbf{p}) = \sum_{i \in \mathcal{B}} \omega_{s_i} \log \left( 1 + \frac{|h_{i,s_i}|^2 p_{s_i}}{\sum_{j \neq i} |h_{i,s_j}|^2 p_{s_j} + \sigma^2} \right) \label{eq:objective} \\ 
&\text{s.t.} \quad \quad 0 \leq p_k \leq P_{\max}, \quad \forall k, \label{eq:constraint_p} \\ 
& \quad\quad \quad s_i \in \mathcal{K}_i \cup \{\varnothing\}, \quad \forall i \in \mathcal{B}. \label{eq:constraint_s} 
\end{align}
\end{subequations}

This is a nonconvex optimization problem with mixed discrete-continuous variables. The discrete scheduling variables play a crucial role in affecting the intercell interference pattern. 
Even for a fixed scheduling, the resulting power control subproblem remains nonconvex. 
This is precisely the difficulty emphasized in the original closed-form FP-based uplink joint scheduling and precoding framework \cite{FPPart2}, where the LDT and the QT are combined to move away from the logarithmic function and decouple scheduling and power variables across cells.

\subsection{SEFP-Based Joint Scheduling and Power Control Algorithm}

In this subsection, we specialize the proposed SEFP framework to the joint uplink scheduling and power control problem.
By applying the RIT in Theorem~\ref{RIT} to the original WSR problem \eqref{eq:main_problem}, $W_{o}(\mathbf{s},\mathbf{p})$ can then be reformulated as \eqref{f_d} on the top of this page. As such, the equivalent problem can be reformulated as
\begin{equation}\label{RIT-reformulate}
\begin{aligned}
\mathop{\text{maximize}}_{\mathbf{s}, \mathbf{p}, \boldsymbol\alpha} \quad & W_r(\mathbf{s}, \mathbf{p}, \boldsymbol\alpha) \\
\text{s.t.} \quad & \eqref{eq:constraint_p}, \eqref{eq:constraint_s},
\end{aligned}
\end{equation}
 where $\boldsymbol\alpha=(\{\alpha_i\}_{i\in\mathcal B})$ denotes the
collection of auxiliary variables introduced by the RIT, with each $\alpha_i$ associated with the $i$-th cell.
\begin{figure*}[!t]
\normalsize 
\begin{equation} \label{f_d}
W_r(\mathbf s,\mathbf p,\boldsymbol\alpha) =\sum_{i\in\mathcal B}  
\frac{\omega_{s_i}(1+\alpha_i)\big[\log(1+\alpha_i)\big]^2 |h_{i,s_i}|^2 p_{s_i}}
{\alpha_i^2 \left( \sum_{j\neq i}|h_{i,s_j}|^2p_{s_j}+\sigma^2 \right) +\big[(1+\alpha_i)\log(1+\alpha_i)-\alpha_i\big] |h_{i,s_i}|^2p_{s_i}}
\end{equation}
\hrulefill 
\vspace*{4pt} 
\end{figure*}
When $(\mathbf{s}, \mathbf{p})$ are fixed, for an active cell with positive desired received signal power, the optimal auxiliary variable $\alpha_i$ is obtained in closed-form as the current SINR of the scheduled user in cell $i$, that is 
\begin{equation} \label{RIT_av_update}
    \alpha_i^\star = \frac{|h_{i,s_i}|^2 p_{s_i}}{\sum_{j \neq i} |h_{i,s_j}|^2 p_{s_j} + \sigma^2}.
\end{equation}
For a zero-SINR cell, the corresponding rate is zero. This is consistent with the boundary case in Theorem~\ref{RIT}, where the associated RIT-reformulated term is treated as a zero contribution, and therefore the auxiliary variable $\alpha_i$ can be assigned any positive value without changing the objective.

As such, substituting $\boldsymbol{\alpha}^\star$ back into $W_r(\mathbf{s}, \mathbf{p}, \boldsymbol{\alpha})$ exactly recovers the original objective function $W_o(\mathbf{s}, \mathbf{p})$. Therefore, the RIT-based formulation \eqref{RIT-reformulate} is equivalent to the original WSR maximization problem \eqref{eq:main_problem}, while moving the SINR term outside the logarithm and preparing the problem for the subsequent QT.
%
It is worth noting that the optimal update rules for $\boldsymbol{\alpha}^\star$ and $\boldsymbol{\gamma}^\star$ in \cite[Eq.25]{FPPart2} are mathematically identical, despite the structural differences between their underlying surrogate functions. This equivalence arises because the first-order derivatives of both surrogate functions coincide, thereby yielding the identical solution during the minimization step.

Next, we apply the QT to the fractional term in the RIT-reformulated objective $W_r(\mathbf s,\mathbf p,\boldsymbol\alpha)$, in order to optimize $(\mathbf s,\mathbf p)$ for fixed $\boldsymbol\alpha$. Let us introduce an auxiliary variable $y_i$ for the $i$-th ratio term in \eqref{f_d}. Using the QT, the RIT objective $W_r(\mathbf{s}, \mathbf{p}, \boldsymbol{\alpha})$ can be further reformulated as \eqref{eq:f_q_combined} on the top of next page, where $b(\cdot)$ and $c(\cdot)$ are defined previously in \eqref{operator_b} and \eqref{operator_c}.
Note that the second equality in \eqref{eq:f_q_combined} is obtained by rearranging the terms associated with the same scheduled user $s_i$.
\begin{figure*}[!t]
\normalsize 
\begin{equation} \label{eq:f_q_combined}
\begin{aligned}
W_{rq}(\mathbf{s}, \mathbf{p}, \boldsymbol{\alpha}, \mathbf{y}) 
& = \sum_{i \in \mathcal{B}} \left( 2y_i \sqrt{\omega_{s_i} c(\alpha_i) |h_{i,s_i}|^2 p_{s_i}} - y_i^2 \left[ \alpha_i^{2} \left( \sum_{j \neq i} |h_{i,s_j}|^2 p_{s_j} + \sigma^2 \right) + b(\alpha_i) |h_{i,s_i}|^2 p_{s_i} \right] \right) \\
& = \sum_{i \in \mathcal{B}} \left( - \alpha_i^{2} y_i^2 \sigma^2 + 2y_i \sqrt{\omega_{s_i} c(\alpha_i) |h_{i,s_i}|^2 p_{s_i}} - b(\alpha_i) y_i^2 |h_{i,s_i}|^2 p_{s_i} - \sum_{j \neq i} \alpha_j^{2} y_j^2 |h_{j,s_i}|^2 p_{s_i} \right).
\end{aligned}
\end{equation}
\hrulefill 
\vspace*{4pt} 
\end{figure*}
In this expression, $\mathbf{y} = (\{y_i\}_{i \in \mathcal{B}})$ denotes the collection of auxiliary variables introduced by the QT. Note that although the notation $y_i$ is the same as that used in the classical FP formulation, its meaning is different here, because the ratio to which the QT is applied is generated by the proposed RIT rather than by the LDT.

Thus, for fixed $\boldsymbol{\alpha}$, the optimization of $W_r(\mathbf{s}, \mathbf{p}, \boldsymbol{\alpha})$ over $(\mathbf{s}, \mathbf{p})$ can be equivalently replaced by the following problem
\begin{equation}\label{RIT-QT-reformulate}
\begin{aligned}
\mathop{\text{maximize}}_{\mathbf{s}, \mathbf{p}, \mathbf{y}} \quad & W_{rq}(\mathbf{s}, \mathbf{p}, \underbrace{\boldsymbol{\alpha}}_{\text{fixed}}, \mathbf{y}) \\
\text{s.t.} \quad & \eqref{eq:constraint_p}, \eqref{eq:constraint_s}.
\end{aligned}
\end{equation}
This reformulation retains the key structural property of the classical FP's reformulated objective function \cite[Eq.(26)]{FPPart2}, after operating LDT and QT successively. Similarly, $W_{rq}$ also separates the desired signal contribution of each scheduled user from its interference penalties imposed on neighboring cells. Consequently, once $\boldsymbol{\alpha}$ and $\mathbf{y}$ are fixed, the scheduling and power-control variables can be updated on a per-cell basis, which enables the distributed scheduling and closed-form power update.

The overall strategy is then to iteratively update the RIT introduced auxiliary variables $\boldsymbol{\alpha}$ according to the current SINR values, and optimize $(\mathbf{s}, \mathbf{p}, \mathbf{y})$ according to the QT formulation above. For a fixed $(\mathbf{s}, \mathbf{p}, \boldsymbol{\alpha})$, the optimal $y_i$ is obtained in closed-form by setting $\partial W_{rq} / \partial y_i=0$ as
\begin{equation} \label{QT_av_update}
    y_i^{\star} = \frac{\sqrt{\omega_{s_i} c(\alpha_i) |h_{i,s_i}|^2 p_{s_i}}}{\alpha_i^{2} \left( \sum_{j \neq i} |h_{i,s_j}|^2 p_{s_j} + \sigma^2 \right) + b(\alpha_i) |h_{i,s_i}|^2 p_{s_i}}.
\end{equation}

Then the problem becomes to optimize $(\mathbf{s}, \mathbf{p})$ for fixed $\boldsymbol{\alpha}$ and $\mathbf{y}$. Consider cell $i$, and suppose that user $k \in \mathcal{K}_i$ is scheduled by BS $i$, i.e., $s_i = k$. Then the terms in $W_{rq}$ that depend on $p_k$ can be collected as
\begin{equation}
2r_i(k)\sqrt{p_k} - d_i(k)p_k,
\end{equation}
where
\begin{equation}
r_i(k) \triangleq y_i \sqrt{\omega_k c(\alpha_i)|h_{i,k}|^2},
\end{equation}
and
\begin{equation}
d_i(k) \triangleq b(\alpha_i)y_i^2|h_{i,k}|^2 + \sum_{j \neq i} \alpha_j^2 y_j^2 |h_{j,k}|^2.
\end{equation}
The coefficient $r_i(k)$ measures the useful link gain of scheduling user $k$ in its serving cell, while the coefficient $d_i(k)$ represents the total penalty associated with this scheduling decision. More specifically, the first term in $d_i(k)$ comes from the desired-link denominator induced by the RIT, whereas the second term accounts for the interference caused by user $k$ to all other cells.

Therefore, if user $k$ is scheduled in cell $i$, its optimal transmit power is obtained by solving
\begin{equation}
\max_{0 \le p_k \le P_{\max}} \quad 2r_i(k)\sqrt{p_k} - d_i(k)p_k.
\end{equation}

The above problem is a simple concave quadratic maximization over $\sqrt{p_k}$. Hence, the optimal transmit power is given in closed form by
\begin{equation}
p_{i,k}^\star = \min \left\{ P_{\max}, \left[ \frac{r_i(k)}{d_i(k)} \right]^2 \right\}.
\end{equation}
Substituting the definitions of $r_i(k)$ and $d_i(k)$, we obtain the explicit update
\begin{equation} \label{update_power_p}
p_{i,k}^{\star} = \min \left\{ P_{\max}, \frac{y_i^2 \omega_k c(\alpha_i) |h_{i,k}|^2}{\left[ b(\alpha_i) y_i^2 |h_{i,k}|^2 + \sum_{j \neq i} \alpha_j^2 y_j^2 |h_{j,k}|^2 \right]^2} \right\}.
\end{equation}
After obtaining the optimal candidate power for each user $k \in \mathcal{K}_i$, we substitute $p_{i,k}^\star$ back into the per-user objective and define the scheduling metric
\begin{equation}
\Xi_i(k) = 2r_i(k)\sqrt{p_{i,k}^\star} - d_i(k)p_{i,k}^\star.
\end{equation}
Equivalently, this metric can be written in a utility-minus-penalty form as
\begin{equation} \label{utility-minus-penalty equ}
\Xi_i(k) = G_i(k) - \sum_{j \neq i} D_j(k),
\end{equation}
where
\begin{equation}
G_i(k) = 2 y_i \sqrt{\omega_k c(\alpha_i) |h_{i,k}|^2 p_{i,k}^\star} - b(\alpha_i) y_i^2 |h_{i,k}|^2 p_{i,k}^\star,
\end{equation}
and
\begin{equation} \label{Djk}
D_j(k) = \alpha_j^2 y_j^2 |h_{j,k}|^2 p_{i,k}^\star.
\end{equation}
Here, $G_i(k)$ denotes the local utility of scheduling user $k$ in cell $i$, while $D_j(k)$ represents the interference penalty imposed on BS $j$. Hence, the scheduling decision has a clear interpretation that each BS selects the user that provides the largest net utility after subtracting the prices paid to other cells.

\begin{remark} \upshape
    Comparing Eq.\eqref{utility-minus-penalty equ}--\eqref{Djk} with Eq.(31)--(33) in \cite{FPPart2} (cf. the utility-minus-penalty metric for user $k$ in the $i$-th cell under the classical FP LDT+QT way), it can be observed our proposed metric preserves the same scheduling structure while modifying three key coefficients induced by the underlying transform, as summarized in Table~\ref{table_comparison}.
\begin{table}[H]
\renewcommand{\arraystretch}{1.3}
\caption{Coefficient Comparison Between Classical FP and Proposed SEFP Scheduling Metrics}
\label{table_comparison}
\centering
\begin{tabular}{p{3.2cm}cc}
\Xhline{0.8pt}
\textbf{Component} & \textbf{\makecell{LDT+QT}} & \textbf{\makecell{RIT+QT}} \\
\hline
Desired-link reward & $1 + \gamma_i$ & $c(\alpha_i)$ \\
Desired-link self penalty & $1$ & $b(\alpha_i)$ \\
Cross-cell interference price & $1$ & $\alpha_j^2$ \\
\Xhline{0.8pt}
\end{tabular}
\end{table}

    Unlike the classical FP metric, where only the desired-link reward depends on the current SINR-related auxiliary variable, the proposed SEFP metric adjusts all three components according to the current auxiliary variables. Therefore, the proposed metric can be interpreted as a more adaptive version than the classical closed-form FP scheduling metric.
    
\end{remark}

Back to the main line, the optimal scheduling decision for cell $i$ under fixed $\boldsymbol{\alpha}$ and $\boldsymbol{y}$ is therefore
\begin{equation} \label{update_schedule_p}
s_i^\star = \begin{cases}
\varnothing, & \text{if } \max\limits_{k \in \mathcal{K}_i} \Xi_i(k) \le 0, \\
\arg\max\limits_{k \in \mathcal{K}_i} \Xi_i(k), & \text{otherwise.}
\end{cases}\end{equation}
If $s_i^\star = \varnothing$, no user is scheduled in cell $i$, and the transmit power in that cell is set to zero. Otherwise, the scheduled user transmits with power
\begin{equation}
p_{s_i^\star} = p_{i,s_i^\star}^\star.
\end{equation}

The entire SEFP approach is summarized in  Algorithm 1.

\begin{algorithm}[H]
\caption{Surrogate-Enhanced FP for Multicell Uplink Joint Scheduling and Power Control}
\label{alg:alg1}
\begin{algorithmic}

\Require Channel coefficients $\{h_{i,k}\}$, user weights $\{\omega_k\}$, maximum transmit power $P_{\max}$, noise power $\sigma^2$, and convergence tolerance $\epsilon$.
\Output Scheduling vector $\mathbf{s}$ and transmit power vector $\mathbf{p}$.

\State \textbf{Initialization}: A feasible scheduling vector $\mathbf{s}^{(0)}$ and power vector $\mathbf{p}^{(0)}$.
\State Set iteration index $t=0$.
\Repeat
    \State Set $t \leftarrow t+1$.
    \For{each BS $i\in\mathcal{B}$}
        \State Update the RIT auxiliary variable $\alpha_i$ by \eqref{RIT_av_update}.
        \State Compute $b(\alpha_i)$, and $c(\alpha_i)$ via \eqref{operator_b} and \eqref{operator_c}.
        \State Update the QT auxiliary variable $y_i$ by \eqref{QT_av_update}.
        \State Update original variables $(s_i,p_{i,k})$ via \eqref{update_schedule_p} and \eqref{update_power_p}.
    \EndFor
\Until{$|W_o(\mathbf{s}^{(t)},\mathbf{p}^{(t)})-W_o(\mathbf{s}^{(t-1)},\mathbf{p}^{(t-1)})|\le \epsilon$}
\State \textbf{return} $\mathbf{s}^{(t)}$ and $\mathbf{p}^{(t)}$.
\end{algorithmic}
\end{algorithm}

\subsection{Convergence and Complexity Analysis}
We summarize the proposed SEFP algorithm's convergence property in the following proposition.
\begin{proposition} \label{pro_converge}
    Algorithm 1 is guaranteed to converge, in the sense that the original WSR objective $W_{o}$ is monotonically non-decreasing after each iteration and eventually upper-bounded at some point.
\end{proposition}
\begin{proof}
See Appendix~\ref{Proof_of_converge}.
\end{proof}

In what follows, we analyze the per-iteration complexity of the proposed SEFP algorithm.
In each iteration, the update of the RIT auxiliary variable $\alpha_i$ requires computing the current SINR of the scheduled user in cell $i$. Since the interference-plus-noise term involves the scheduled users in the other cells, updating all $\alpha_i$ requires $O(|\mathcal{B}|^2)$ scalar operations. The additional computation of $b(\alpha_i)$, and $c(\alpha_i)$ only requires $O(|\mathcal{B}|)$ operations. Similarly, the QT auxiliary variable update has the same interference summation structure, and therefore updating all $y_i$'s also requires $O(|\mathcal{B}|^2)$ operations.
The dominant cost comes from the closed-form scheduling and power-control update. For each candidate user $k \in \mathcal{K}_i$, computing the corresponding interference penalty coefficient requires summing over the other BSs, which costs $O(|\mathcal{B}|)$ operations. Once this coefficient is obtained, the candidate transmit power and the scheduling metric are computed in closed-form with constant additional complexity. Therefore, evaluating all $|\mathcal{K}_i|$ candidate users in one cell requires $O(|\mathcal{K}_i| |\mathcal{B}|)$ operations, and evaluating all cells requires $O(|\mathcal{K}_i| |\mathcal{B}|^2)$ operations.

The proposed SEFP algorithm has the same asymptotic per-iteration complexity as the classical closed-form FP algorithm \cite[Sec. VI-B]{FPPart2}. Both algorithms preserve the per-cell separable scheduling structure and admit closed form candidate power updates. The difference is that SEFP replaces the LDT-induced coefficients in the classical FP method with the RIT-induced coefficients $b(\alpha_i)$ and $c(\alpha_i)$. These extra scalar computations introduce only a constant factor overhead and do not change the complexity order.

\section{Simulation results}
In this section, we evaluate the proposed SEFP algorithm across different network utility metrics and SNRs. We benchmark performance against the following three baselines.\footnote{The source code for the three baseline algorithms is available online at \baselinecode}
\begin{itemize}
    \item {\textit{FP}}: This baseline corresponds to the original FP-based coordinated uplink scheduling and power control algorithm in \cite{FPPart2}. It first applies the LDT to move the SINR term outside the logarithm and then applies the QT to obtain a separable scheduling and power control structure. This baseline is included to directly examine the benefit of replacing the LDT-induced fractional surrogate with the proposed RIT-induced one.
    \item {\textit{Power Control}}\footnote{Employing a power control algorithm to solve optimization problems involving scheduling variables suffers from a severe limitation. Specifically, it frequently results in scenarios where certain users are never scheduled during the entire transmission horizon—a phenomenon referred to as \textbf{premature turning-off} \cite[Sec. IV-B]{FPPart2}.}: This baseline treats uplink scheduling implicitly as a global power control problem over all users. Users assigned zero transmit power are regarded as unscheduled, while users assigned positive powers are scheduled. Since the resulting power control problem is highly nonconvex, the WMMSE algorithm \cite{WMMSE_1,WMMSE_2} is used to obtain a local optimum. This baseline reflects the performance of using continuous power optimization to indirectly determine scheduling decisions.
    \item {\textit{Fixed Interference Method}}: This baseline alternates between scheduling and power control under a fixed-interference approximation. In the scheduling stage, each cell selects the user that maximizes its weighted rate while treating the interference from the previous iteration as fixed. In the power control stage, the transmit powers of the scheduled users are updated by solving the corresponding weighted sum-rate maximization problem. These two steps repeat until convergence or until the prescribed maximum number of iterations is reached.
\end{itemize}

The remainder of this section is organized as follows: Section~\ref{Section5SubsectionA} outlines the common experimental configurations. Sections~\ref{Section5SubsectionB} and \ref{Section5SubsectionC} evaluate the performance of the proposed algorithm against baseline approaches under three distinct network utility metrics and various received SNR levels, respectively. 

\subsection{Experimental Setup} \label{Section5SubsectionA}
We consider an uplink multi-cell SISO network with the standard $L=7$ cells with topology wrap-around. The inter-BS distance is set to 0.8 km, corresponding to a cell radius of $R_c = 0.8 / \sqrt{3}$ km. The channel incorporates both large-scale fading and small-scale fading. The large-scale fading consists of a distance-dependent pathloss component and a log-normal shadowing component. Specifically, the pathloss is modeled as $128.1 + 37.6 \log_{10}(\textrm{d})$ dB, where $\textrm{d}$ is the link distance in kilometers, and the shadowing follows a zero-mean Gaussian distribution with a standard deviation of 8 dB. The small-scale fading is characterized by the Rayleigh fading model, where the channel coefficients are independent and identically distributed (i.i.d.) zero-mean circularly symmetric complex Gaussian (CSCG) random variables with unit variance.

The background noise spectral density is $-169\ \textrm{dBm/Hz}$ over a 10 MHz bandwidth. 
The maximum transmit power spectral density of each user is set to $-47\ \textrm{dBm/Hz}$.
For each independent channel realization, a total of 84 users are uniformly distributed in the network, and each user is associated with the BS from which it receives the strongest large-scale channel gain. To ensure that all BSs are utilized in practice, we enforce a constraint that every BS must be associated with at least one user; otherwise, the user deployment is regenerated until this condition is satisfied. In each scheduling interval, at most one user is scheduled in each cell, and the transmit power of the scheduled user is optimized jointly with the scheduling decision. In the following experiments, one channel realization corresponds to one Monte Carlo experiment, and one scheduling interval corresponds to one time slot. For fairness, all algorithms are initialized with the same feasible scheduling and power vectors for each channel realization.

In all simulations, the maximum iteration number of each iterative algorithm is set to 50, which is empirically sufficient for convergence under the considered settings. In addition, the proposed SEFP algorithm is equipped with an additional convergence tolerance stopping criterion described in Algorithm 1, where the convergence tolerance is set to $10^{-10}$. Thus, SEFP stops when either the tolerance criterion is met or the maximum iteration number is reached.

\subsection{Simulation across Different Network Utility Metrics} \label{Section5SubsectionB}
In this subsection, we evaluate the algorithm's performance under three distinct network utility metrics.
\subsubsection{\textbf{Proportional-Fairness (PF) Utility Metric}}\ 
\newline
\indent Given a specific channel realization, the weight assigned to user $k$ in time slot $t$ is updated inversely proportional to its long-term average throughput, expressed as $$\omega_k(t) = {1}/{\bar{R}_k(t)},$$ where $\bar{R}_k(t)$ represents the user $k$'s historical average throughput before the time slot $t$. Consequently, a WSR optimization subproblem is solved at each time slot. The PF metric represents the long-term ($T$ scheduling intervals) network log-utility
$$U_{\text{PF}}(T) = \sum_{k \in \mathcal{K}} \log(\bar{R}_k(T)).$$

Notably, the simulation setting here is identical to that used in \cite[Sec.~IV-D]{FPPart2}, so that the performance gain can be attributed to the proposed Reciprocal-Inversion transform rather than to a different network configuration.
\begin{figure}[htbp] 
  \centering 
  \hspace*{-0.4cm}\includegraphics[scale=0.66]{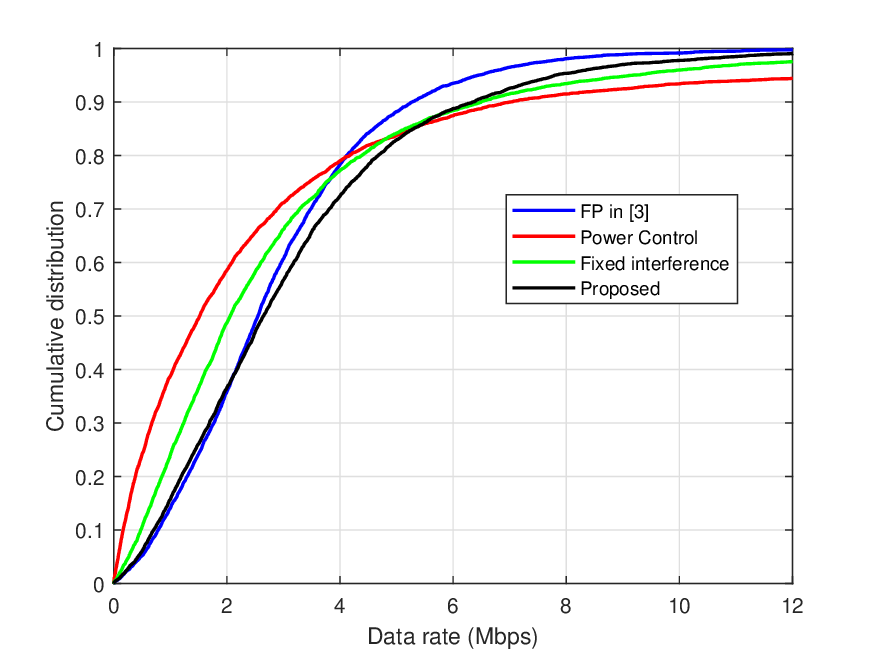} 
  \caption{CDF of the long-term average user rates under the PF utility metric with $60$ Monte Carlo channel realizations and $500$ scheduling intervals per realization.} 
  \label{fig:cdf_60MC_500TS} 
\end{figure}

\begin{table}[!t] 
\centering 
\caption{PF utility comparison under the same setting in \cite{FPPart2}} \label{tab:pf_utility_comparison} 
\small 
\begin{tabular}{c|c|c} \hline Algorithm & Mean  & Std.  \\ \hline FP in \cite{FPPart2} & 65.946 & 8.038 \\ Power Control & 17.436 & 35.686 \\ Fixed Interference & 51.112 & 9.197 \\ Proposed & 68.943 & 8.199 \\ \hline \multicolumn{3}{c}{Paired comparison: Proposed versus FP in \cite{FPPart2}} \\ \hline Mean difference & Std. difference & Proposed $>$ FP in [3]\\ \hline 2.996 & 0.804 & $60/60$ \\ \hline \multicolumn{3}{c}{Minimum and maximum paired differences: $0.489$ and $4.475$} \\ \hline 
\end{tabular} 
\begin{center}
\footnotesize
For the Power Control baseline, a small numerical floor is applied to the long-term average rate before taking the logarithm. Specifically, we compute the PF utility using $\max{(\bar R_k(T),10^{-12}})$ in place of $\bar R_k(T)$ for every user $k$. Because under some specific channel realizations, the premature turning-off phenomena (mentioned in footnote 3) occur on certain users, so that their long-term average rates become zero and $U_{\mathrm{PF}}(T)$ would consequently be $-\infty$. This floor does not eliminate the penalty for unscheduled users; rather, each such user contributes $\log(10^{-12})\approx -27.63$ to the system PF metric, thereby maintaining a severe penalty while keeping the evaluation finite.
\end{center}
\end{table}

Fig.~\ref{fig:cdf_60MC_500TS} shows the empirical cumulative distribution function (CDF) for each user's average data rate $\bar{R}_k(T)$. For each channel realization, the number of time slots is $T=500$, and the channel realization number is 60. Our proposed locates more to the right among the four algorithms, especially in the medium-rate region. Table \ref{tab:pf_utility_comparison} shows our proposed algorithm outperforming others more apparently, which achieves the highest mean PF utility among all compared methods, improving the mean utility from $65.946$ for FP to $68.943$. To further verify whether the improvement over FP is statistically stable, we perform a paired comparison between the proposed SEFP algorithm and FP over the same $60$ independent channel realizations. The average paired utility gain is $2.996$, and the corresponding approximate $95\%$ confidence interval (CI) is $[2.793,3.200]$, which is strictly positive. Moreover, the proposed algorithm outperforms FP in all $60$ channel realizations, with the minimum paired gain still being $0.489$. These results demonstrate that the gain of the proposed SEFP algorithm under the PF utility metric is not caused by a few favorable channel realizations, but is consistently observed across the whole Monte Carlo experiment.

\subsubsection{\textbf{Equal-Weight Sum-Rate (EW-SR) Metric}}\ 
\newline
\indent Under this metric, for each channel realization, all the users' weights are set as 
$$\omega_k = 1, \quad \forall k.$$
We care about the network's EW-SR, expressed as
\begin{equation}
 U_{\text{EW-SR}} = \sum_{k \in \mathcal{K}} {R}_k.  
\end{equation}
Notice that this metric is instantaneous, which implies that the time slot number $T=1$ naturally. 

Fig.~\ref{fig:equalweight_SR_10000MC} presents the empirical CDF of the network EW-SR over 10,000 independent Monte Carlo channel realizations. Since the EW-SR metric assigns the same priority to all users, the resulting performance directly reflects the capability of each algorithm to maximize the instantaneous network throughput under the same channel realization. It can be observed that the proposed SEFP algorithm provides a consistently right-shifted CDF compared with the three baseline algorithms.
Table \ref{tab:best_algorithm_equal_weight_wsr} further reports the four algorithms' mean, standard deviation, and best count percentage (BCP)  under the EW-SR metric. The best count percentage denotes the percentage of Monte Carlo realizations in which an algorithm achieves the maximum EW-SR. If multiple algorithms attain the same maximum EW-SR value in one realization, all tied algorithms are counted. The proposed algorithm achieves the maximum EW-SR in 9,596 out of 10,000 Monte Carlo realizations. In contrast, FP in [3], Power Control, and Fixed Interference achieve the best EW-SR in only 0.44\%, 1.63\%, and 2.11\% of the realizations, respectively. The dominant best-count percentage of the proposed algorithm clearly demonstrates our proposed algorithm's superiority.

\begin{figure}[htbp] 
  \centering 
  \hspace*{-0.4cm}\includegraphics[scale=0.66]{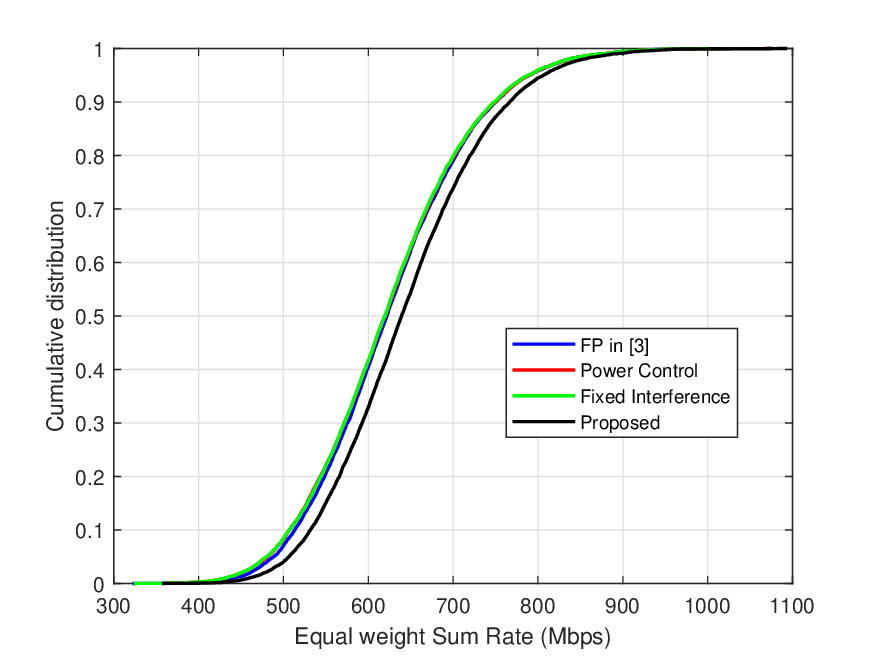} 
  \caption{CDF of the network EW-SR utility with $10,000$ Monte Carlo channel realizations.} 
  \label{fig:equalweight_SR_10000MC} 
\end{figure}

\begin{table}[t]
\caption{Mean, Standard Deviation (Std.) and BCP Statistics for the EW-SR Metric}
\label{tab:best_algorithm_equal_weight_wsr}
\small
\centering
\begin{tabular}{lccc}
\hline
\textbf{Algorithm} & \textbf{Mean} & \textbf{Std.} & \textbf{BCP} \\
\hline
FP in \cite{FPPart2} & 627.730 & 92.582 & 0.44\% \\
Power Control & 624.898 & 94.540 & 1.63\% \\
Fixed Interference & 624.503 & 94.276 & 2.11\% \\
Proposed & 645.149 & 92.084 & 95.96\% \\
\hline
\end{tabular}
\end{table}

\subsubsection{\textbf{Normalized Random-Priority Weighted Sum-Rate (NRP-WSR) Metric}}\ 
\newline
\indent The EW-SR metric evaluates the instantaneous network throughput when all users are assigned the same priority. However, in many scheduling scenarios, different users may have different service priorities. To examine the robustness of the proposed algorithm under heterogeneous user priorities, we further consider an NRP-WSR metric.

For each channel realization, random user weights are generated and then normalized on a per-cell basis. Specifically, the weights of the users associated with the BS-$i$ satisfy
$$\sum_{k \in \mathcal{K}_i} \omega_k = \frac{|\mathcal{K}|}{|\mathcal{B}|}=12, \quad \forall i \in \mathcal{B}.$$
The resulting metric emphasizes the effect of random user priorities while avoiding bias caused by unequal user association among BSs. The corresponding NRP-WSR is defined as
\begin{equation}
 U_{\text{NRP-WSR}} = \sum_{k \in \mathcal{K}} \omega_k{R}_k.  
\end{equation}

For each of the 500 independent Monte Carlo channel realizations, we generate 20 independent normalized random-priority weight settings. Therefore, a total of 10,000 NRP-WSR samples are obtained for each algorithm. We believe that somehow the results characterize if our proposed algorithm could achieve the more external rim points of the WSR achievable rate region of the scheme combining power control and scheduling at the transmitter as well as treating interference as noise at the receiver.

Fig.~\ref{fig:randomweightWSR} shows the empirical CDF of the NRP-WSR achieved by the four algorithms. The Power Control baseline has the most left-shifted CDF, which suggests that treating scheduling implicitly through continuous power control is less effective when user priorities vary randomly. The proposed SEFP algorithm still exhibits a clear right shift compared with the three baselines, which shows our proposed algorithm's superiority under the NRP-WSR metric. Table \ref{tab:best_algorithm_NPR-WSRMETRIC} further summarizes the algorithms' mean, standard deviation, and BCP of the NRP-WSR metric. The proposed SEFP algorithm achieves the highest mean NRP-WSR, outperforming others by at least 2.47\%, and attains the best NRP-WSR in 78.18\% of all channel-weight realizations.

Overall, the simulation results delineate that our proposed algorithm is not only the best under the PF metric (the metric considered in \cite{FPPart2}), but also the best under the EW-SR and NRP-WSR metrics, where the original FP algorithm performs relatively poorly.

\begin{table}[h]
\caption{Mean, Standard Deviation (Std.) and BCP for the NRP-WSR Metric}
\label{tab:best_algorithm_NPR-WSRMETRIC}
\small
\centering
\begin{tabular}{lccc}
\hline
\textbf{Algorithm} & \textbf{Mean} & \textbf{Std.} & \textbf{BCP} \\ \hline
FP in [3]          & 826.479       & 163.526       & 1.19\%       \\
Power Control      & 757.175       & 175.354       & 3.74\%       \\
Fixed Interference & 831.948       & 163.906       & 16.93\%      \\
Proposed           & 853.001       & 163.836       & 78.18\%      \\ \hline
\end{tabular}
\end{table}

\begin{figure}[htbp] 
  \centering 
  \hspace*{-0.4cm}\includegraphics[scale=0.66]{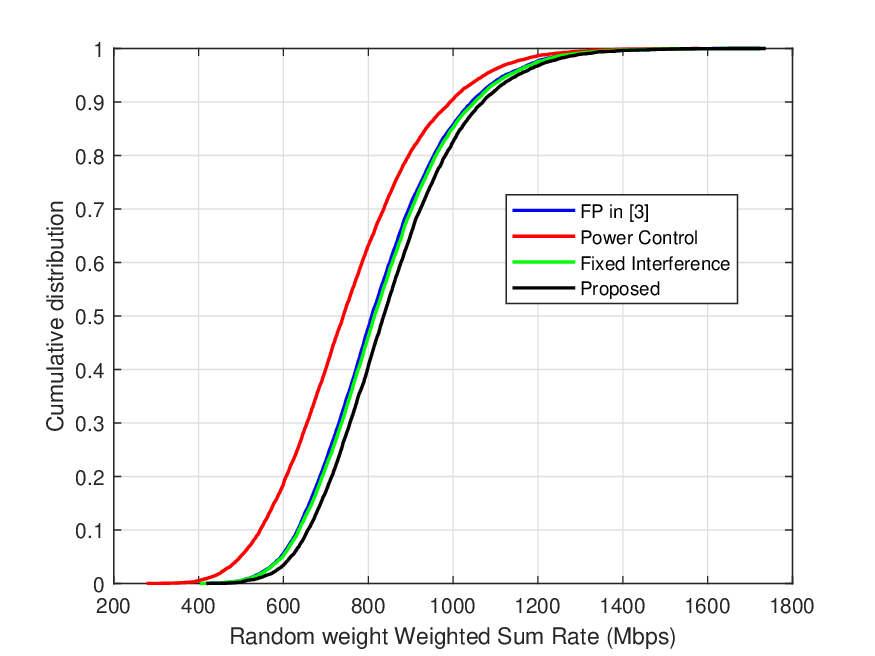} 
  \caption{CDF of the network NRP-WSR utility with $500$ Monte Carlo channel realizations and $20$ normalized random-priority weight settings per realization.} 
  \label{fig:randomweightWSR} 
\end{figure}

\subsection{Simulation across Different SNRs} \label{Section5SubsectionC}
In this subsection, we further evaluate the performance of the proposed SEFP algorithm and the baseline algorithms under different SNR levels. The nominal SNR is varied from 0 dB to 14 dB with a step size of 2 dB. For a fair comparison, the background noise power, channel generation procedure, user association rule, and network topology are kept unchanged for all SNR points. The SNR variation is implemented by adjusting the maximum transmit power constraint $P_{\max}$ of each user.

Specifically, the SNR is defined as a nominal reference-link SNR at distance $0.4$ km from the serving BS. According to the pathloss model in Section \ref{Section5SubsectionA}, the corresponding pathloss is approximately $113.14$ dB. Since the noise spectral density is fixed at $-169$ dBm/Hz over a $10$ MHz bandwidth, the total noise power is $-99$ dBm. Therefore, for a target nominal SNR $\rho$ in dB, the maximum transmit power is set as
$$
P_{\max,\text{dBm}}(\rho) = \rho + \text{PL}(d_0) + N_{\text{dBm}} \approx \rho + 14.14.
$$
Thus, each $2$ dB increase in the nominal SNR corresponds to a $2$ dB increase in the user maximum transmit power constraint.

\begin{figure}[htbp] 
  \centering 
  \hspace*{-0.4cm}\includegraphics[scale=0.66]{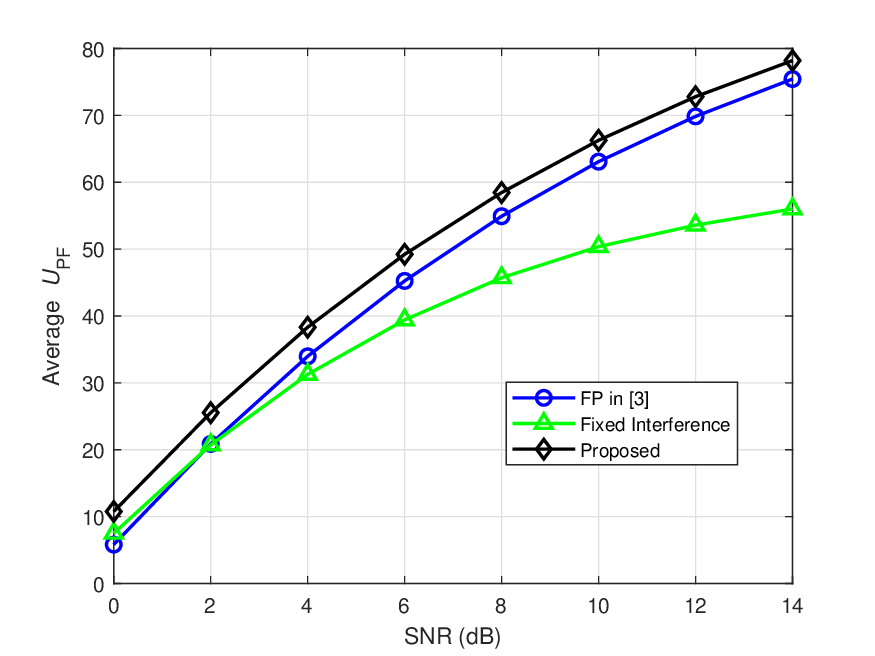} 
  \caption{Average PF utility vs SNR from 0 to 14 under 100 Monte Carlo channel realizations.} 
  \label{fig:SNRchange} 
\end{figure}

For each SNR point, we perform 100 independent Monte Carlo channel realizations, and each realization contains 500 scheduling intervals under the PF metric. Fig. \ref{fig:SNRchange} reports the average PF utility achieved by FP, the Fixed Interference Method, and the proposed SEFP algorithm. The Power Control baseline is omitted from this comparison because it frequently suffers from the premature turning-off phenomenon in the PF setting, where some users may never be scheduled over the whole transmission horizon. This leads to extremely low PF utility and makes it unsuitable as a meaningful PF benchmark.

As shown in Fig. \ref{fig:SNRchange}, the average PF utility of all algorithms increases with the nominal SNR. More notably, the proposed SEFP algorithm consistently achieves the highest PF utility across the entire SNR range. This confirms that the gain of the proposed RIT-based reformulation is not limited to a specific SNR regime. Compared with the Fixed Interference Method, the advantage of the SEFP becomes more evident at medium and high SNRs, where intercell interference plays a more dominant role and the fixed-interference approximation becomes less accurate. Compared with the original FP algorithm in \cite{FPPart2}, the SEFP also maintains a clear and stable performance gain, demonstrating the benefit of replacing the LDT-induced surrogate with the proposed RIT-induced surrogate. Notably, for every tested SNR point, the SEFP attains the highest PF utility in every individual Monte Carlo realization among the compared algorithms. This observation further confirms that the performance gain of the SEFP is not caused by averaging over a few favorable channel realizations, but is consistently observed under all tested SNR conditions.

\section{Conclusion and Discussion}
This paper investigated the joint uplink scheduling and power control problem in a coordinated multicell wireless network. By revisiting the classical LDT from an MM perspective, we showed that its reciprocal-coordinate surrogate can be conservative. Motivated by this observation, we proposed the RIT, which constructs a first-order tight lower bound for the logarithmic rate function while preserving a QT-compatible fractional structure.
Further, we developed the SEFP algorithm, which retains the per-cell separability and closed-form update structure of the classical FP framework, while introducing SINR-adaptive coefficients into the scheduling and power-control metric. The resulting algorithm monotonically improves the original WSR objective and has the same asymptotic per-iteration complexity as the classical closed-form FP algorithm.
Simulation results under the different WSR metrics and different SNR levels demonstrated that the sefp consistently outperforms the classical FP method and other baselines. Future work includes the extensions of the proposed RIT-based framework to MIMO beamforming and other settings that the classical FP well-applied.

\appendices
\section{Proof of Proposition 1}
\begin{enumerate}
    \item By evaluating the LDT surrogate at $r = 0$ and using the elementary inequality $\log(1 + x) < x$ for $x > 0$, we obtain
    $$\ell_{\text{LDT}}(0; \bar{r}) < 0 = f(0), \quad \bar{r} > 0.$$
    Thus, the LDT surrogate fails to match the boundary value of $f(r)$ at $r = 0$.
    \item The second derivatives of $\ell_{\text{LDT}}(r; \bar{r})$ and $f(r)$ with respect to $r$ are
    $$f''(r) = -\frac{1}{(1 + r)^2}, \quad \ell_{\text{LDT}}''(r; \bar{r}) = -\frac{2(1 + \bar{r})}{(1 + r)^3}.$$
    Thus,
    $$\ell_{\text{LDT}}''(\bar{r}; \bar{r}) = 2f''(\bar{r}). $$
    Since $f''(\bar{r}) < 0$, the LDT surrogate is locally more curved, and hence more conservative, around the touching point.
    \item As $r \to \infty$, $f(r)=\log (1+r)\to \infty$ whereas $\ell_{\text{LDT}}(r; \bar{r})\to \log(1 + \bar{r}) + 1 < \infty.$
    This proves that the LDT surrogate saturates to a finite value in the high-ratio regime, whereas $f(r)$ grows unboundedly. 

\end{enumerate}

\section{Proof of Corollary 1} \label{proof_of_corr1}
The surrogate $\ell_{\text{RIT}}(r; \bar{r})$ is well-defined for every $r \ge 0$ and $\bar{r} > 0$. Moreover, it preserves the exact boundary value of $f(r)$ at $r = 0$, that is
\begin{equation}
 \ell_{\text{RIT}}(r=0; \bar{r}) = 0 = \log(1 + 0)=f(r=0).   
\end{equation}
It is also first-order tight at $r = \bar{r}$, since
\begin{equation}
\ell_{\text{RIT}}(r =\bar{r}; \bar{r}) = \log(1 + \bar{r})=f(r=\bar{r}),
\end{equation}
and
\begin{equation}
\left. \frac{\partial \ell_{\mathrm{RIT}}(r; \bar r)}{\partial r} \right|_{r=\bar r} = \frac{1}{1 + \bar r} = \left. \frac{\partial f( r)}{\partial r} \right|_{r=\bar r}.
\end{equation}
Therefore, $\ell_{\text{RIT}}(r ; \bar{r})$ is a valid MM minorization surrogate for $f(r)$.

\section{Proof of Proposition~\ref{RITtightproposition}} \label{proof_of_RITtight}
Since an MM surrogate can be obtained by fixing the auxiliary variables of the corresponding equivalent transform at their optimal values for the current iterate, it inherits the identical functional structure as the transform objective. Moreover, both the LDT+QT objective $F_{\ell q}(\mathbf{x}, \boldsymbol{\gamma}, \mathbf{y})$ in \eqref{Flq} and the RIT+QT objective $F_{rq}(\mathbf{x}, \boldsymbol{\alpha}, \mathbf{y})$ in \eqref{Frq} are separable across $m$. Therefore, it is sufficient to compare the two MM surrogates corresponding to the $m$-th sub-function.

At the $t$-th MM iteration, let $\mathbf{x}^{(t)}$ denote the current iterate. By choosing the reference point as the current value $\bar r_m^{(t)}$ defined in \eqref{ratio for t-th iter}, the resulting LDT+QT surrogate $\ell_{m}^{\text{LQ}}$ of $f(r_m(\mathbf{x}))$ can be expressed as
\begin{equation} \label{l_LQ} 
\begin{aligned}
    &\ell_{m}^{\text{LQ}}(\mathbf{x}; \mathbf{x}^{(t)}) = \\&\log(1 + \bar{r}^{(t)}_m) - \bar{r}^{(t)}_m + 2\sqrt{\bar{r}^{(t)}_mr_m(\mathbf{x})} - \frac{\bar{r}^{(t)}_m}{1 + \bar{r}^{(t)}_m}(1 + r_m(\mathbf{x})).
\end{aligned}
\end{equation}
The resulting RIT+QT surrogate $\ell_{m}^{\text{RQ}}$ of $f(r_m(\mathbf{x}))$ can be expressed as
\begin{equation} \label{l_RQ}
\ell_{m}^{\text{RQ}}(\mathbf{x}; \mathbf{x}^{(t)}) = 2 \log(1 + \bar{r}_m) \sqrt{\frac{r_m(\mathbf{x})}{\bar{r}_m}} - \frac{\bar{r}_m}{1 + \bar{r}_m} - \frac{b(\bar{r}_m)r_m(\mathbf{x})}{\bar{r}_m(1 + \bar{r}_m)}.
\end{equation}

Taking the difference between \eqref{l_RQ} and \eqref{l_LQ}, we have 
\begin{equation}
\ell_{m}^{\text{RQ}}(\mathbf{x}; \mathbf{x}^{(t)}) - \ell_{m}^{\text{LQ}}(\mathbf{x}; \mathbf{x}^{(t)}) = \big[\bar{r}_m - \log(1 + \bar{r}_m)\big] \big(\sqrt{\frac{r_m(\mathbf{x})}{\bar{r}_m}} - 1\big)^2.
\end{equation}
Since
\begin{equation}
\bar{r}_m - \log(1 + \bar{r}_m) > 0, \quad \bar{r}_m > 0,
\end{equation}
we have
\begin{equation}
\ell_{m}^{\text{RQ}}(\mathbf{x}; \mathbf{x}^{(t)}) \geq \ell_{m}^{\text{LQ}}(\mathbf{x}; \mathbf{x}^{(t)}).
\end{equation}
Then, since $\omega_m\ge0$, multiplying the above point-wise inequality by $\omega_m$ and summing over all $m=1,...,M$ yields
\begin{equation}
\sum_{m=1}^{M} \omega_m \ell_{m}^{\text{RQ}}(\mathbf{x}; \mathbf{x}^{(t)})
\geq
\sum_{m=1}^{M} \omega_m \ell_{m}^{\text{LQ}}(\mathbf{x}; \mathbf{x}^{(t)}).
\end{equation}
This proves Proposition~\ref{RITtightproposition}.

\section{Proof of Proposition~\ref{pro_converge}}\label{Proof_of_converge}
The proof relies on the tightness of the RIT and the QT. For fixed $(\mathbf{s},\mathbf{p})$, the RIT gives
$$
W_o(\mathbf{s},\mathbf{p}) = \max_{\boldsymbol\alpha} W_r(\mathbf s,\mathbf p,\boldsymbol\alpha),
$$
where the maximum is attained by the update in \eqref{RIT_av_update}. For fixed $(\mathbf s,\mathbf p,\boldsymbol\alpha)$, the QT gives
$$
W_r(\mathbf s,\mathbf p,\boldsymbol\alpha) = \max_{\mathbf y} W_{rq}(\mathbf s,\mathbf p,\boldsymbol\alpha,\mathbf y),
$$
where the maximum is attained by the update in \eqref{QT_av_update}.

Let $(\mathbf s^{(t)}, \mathbf p^{(t)})$ be the variables at the end of the $t$-th iteration, and let $\boldsymbol\alpha^{(t)}$ and $\mathbf y^{(t)}$ be the corresponding optimal auxiliary variables. Also let $\bar{\mathbf y}^{(t)}$ be the optimal QT auxiliary variable evaluated at $(\mathbf s^{(t+1)}, \mathbf p^{(t+1)}, \boldsymbol\alpha^{(t)})$. Then,
\begin{align*}
W_o(\mathbf s^{(t+1)}, \mathbf p^{(t+1)}) 
&= W_r(\mathbf s^{(t+1)}, \mathbf p^{(t+1)}, \boldsymbol\alpha^{(t+1)}) \\
&\geq W_r(\mathbf s^{(t+1)}, \mathbf p^{(t+1)}, \boldsymbol\alpha^{(t)}) \\
&= W_{rq}(\mathbf s^{(t+1)}, \mathbf p^{(t+1)}, \boldsymbol\alpha^{(t)}, \bar{\mathbf y}^{(t)}) \\
&\geq W_{rq}(\mathbf s^{(t+1)}, \mathbf p^{(t+1)},\boldsymbol \alpha^{(t)}, \mathbf y^{(t)}) \\
&\geq W_{rq}(\mathbf s^{(t)}, \mathbf p^{(t)}, \boldsymbol\alpha^{(t)}, \mathbf y^{(t)}) \\
&= W_r(\mathbf s^{(t)}, \mathbf p^{(t)}, \boldsymbol\alpha^{(t)}) \\
&= W_o(\mathbf s^{(t)}, \mathbf p^{(t)}).
\end{align*}
Therefore, $W_o$ is monotonically non-decreasing after each iteration. Since the value of $W_o$ (represents the network WSR) is bounded above, the convergence is then proved.

\newpage

\vfill

\end{document}